%% file: Stochastic_Gradient_Coding_v2.tex
\documentclass[11pt,drafunctiontcls,onecolumn]{IEEEtran}
\include{mary_macros}

\usepackage{%
	amsmath,%
	accents,%
	amssymb,%
	amsthm,%
	graphicx,%
	array,%
	bchart,%
	caption,%
	cite,%
	diagbox,%
	epsfig,%
	latexsym,%
	multirow,%
	pgfplots,%
	sidecap,%
	subcaption,
	tikz,%
	varwidth,%
	wrapfig,%
	xcolor,%
	arydshln,%
	slashbox,%
	mathtools%
}

\usepackage[inline]{enumitem}
\usepackage[linesnumbered,ruled,vlined]{algorithm2e}

\sidecaptionvpos{figure}{c}

\usetikzlibrary{%
	arrows,%
	automata,%
	calc,%
	patterns,%
	plotmarks,%
	positioning,%
	shapes.geometric,%
	shapes,%
	snakes,%
}

\pgfplotsset{
    legend image with text/.style={
        legend image code/.code={%
            \node[anchor=center] at (0.3cm,0cm) {#1};
        }
    },
}

\def\erasurehead{{\sc ErasureHead}}

\def \bes{\begin{equation*}}
\def \ees{\end{equation*}}
\def \bas{\begin{align*}}
\def \eas{\end{align*}}
\def \be{\begin{equation}}
\def \ee{\end{equation}}
\def \bbm{\begin{bmatrix}}
\def \ebm{\end{bmatrix}}

\newcommand{\grad}[2]{ \nabla \mathcal{L}(\mathbf{a}_#1, \mathbf{\bbeta}_#2)}
\def \ghatt {\hat{\mathbf{g}}_t}

\def \tZ {\tilde{Z}}

\def \nbiter{\nu}

\newcommand{\rawad}[1]{{\footnotesize  [\textbf{\textcolor{green!50!black}{#1}} \textcolor{green!30!black}{--rawad}]\normalsize}}

\definecolor{bleudefrance}{rgb}{0.19, 0.55, 0.91}

\def \w{{\fontfamily{cmtt}\selectfont \text{W}}}

\newcolumntype{C}[1]{>{\centering\let\newline\\\arraybackslash\hspace{0pt}}m{#1}}

\usetikzlibrary{arrows,positioning,automata,calc}

\newtheorem*{theorem*}{Theorem}

\IEEEoverridecommandlockouts

\begin{document}

\newlength\figureheight
\newlength\figurewidth

\title{Stochastic Gradient Coding for  Straggler Mitigation in Distributed Learning}
\author{
\IEEEauthorblockN{Rawad Bitar\IEEEauthorrefmark{1}, Mary Wootters\IEEEauthorrefmark{2}, and Salim El Rouayheb\IEEEauthorrefmark{1}}\\
\IEEEauthorblockA{\small\IEEEauthorrefmark{1}ECE department, Rutgers University} \and 
\IEEEauthorblockA{\small\IEEEauthorrefmark{2}CS and EE departments, Stanford University}\\
\IEEEauthorblockA{\small rawad.bitar@rutgers.edu, marykw@stanford.edu, salim.elrouayheb@rutgers.edu}
 \thanks{
M. Wootters is partially funded by NSF grant CCF-1657049 and by NSF CAREER grant CCF-1844628. R. Bitar and S. El Rouayheb  were partially funded by NSF grant  CNS-1801708.} 
}

\maketitle

\begin{abstract}
We consider distributed gradient descent in the presence of stragglers.  Recent work on \em gradient coding \em and \em approximate gradient coding \em have shown how to add redundancy in distributed gradient descent to guarantee convergence even if some workers are \em stragglers\em---that is, slow or non-responsive.  In this work we propose an approximate gradient coding scheme called \em Stochastic Gradient Coding \em (SGC), which works when the stragglers are random.  SGC distributes data points redundantly to workers according to a pair-wise balanced design, and then simply ignores the stragglers.  We prove that the convergence rate of SGC mirrors that of batched Stochastic Gradient Descent (SGD) for the $\ell_2$ loss function, and show how the convergence rate can improve with the redundancy.  We also provide bounds for more general convex loss functions.  We show empirically that SGC requires a small amount of redundancy to handle a large number of stragglers and that it can outperform existing approximate gradient codes when the number of stragglers is large.
\end{abstract}

\section{Introduction}
\input{intro}

\section{Setup}\label{sec:model}
\input{model}

\section{Stochastic gradient coding}\label{sec:SGC}
\input{SGC}

\section{Summary of our main results}\label{sec:results}
\input{results_summary}

\section{Theoretical Results} \label{sec:convergence}
\input{convergence}

\section{Simulation Results}\label{sec:sims}
\input{sims}

\section{Proofs}\label{sec:proofs}
\subsection{Proof of Theorem~\ref{thm:convl2}}
\input{wtd_proof}
\subsection{Proof of Theorem~\ref{thm:conv}}
\input{analysis}

\section{Related Work}\label{sec:related}

\input{related}

\section*{Acknowledgements} We thank Deanna Needell for helpful pointers to the literature.

\bibliographystyle{ieeetr}
\bibliography{CodedComp}

\end{document}

%% file: mary_macros.tex
\usepackage[margin=1.125in]{geometry}
\usepackage{xcolor}
\definecolor{DarkGreen}{rgb}{0.1,0.5,0.1}
\definecolor{DarkRed}{rgb}{0.5,0.1,0.1}
\definecolor{DarkBlue}{rgb}{0.1,0.1,0.5}
\usepackage{bm}

\usepackage[small]{caption}
\usepackage[pdftex]{hyperref}
\hypersetup{
    unicode=false,          
    pdftoolbar=true,        
    pdfmenubar=true,        
    pdffitwindow=false,      
    pdfnewwindow=true,      
    colorlinks=true,       
    linkcolor=DarkBlue,          
    citecolor=DarkGreen,        
    filecolor=DarkGreen,      
    urlcolor=DarkBlue,          
    pdftitle={},
    pdfauthor={},    
}
\usepackage{hyperref,cite}
\usepackage{amssymb,amsmath,amsthm,amsfonts}
\usepackage{fullpage,nicefrac,comment}
\usepackage{tikz}
\usetikzlibrary{arrows,decorations.pathmorphing,decorations.shapes,snakes,patterns,positioning}

\newcommand{\ba}{\bm{a}}
\newcommand{\bx}{\bm{x}}
\newcommand{\by}{\bm{y}}
\newcommand{\bg}{\bm{g}}

\newcommand{\bbeta}{\bm{\beta}}
\newcommand{\tgt}{\tilde{\gamma}_t}

\newcommand{\EE}{\mathbb{E}}

\newcommand{\norm}[1]{\left\|#1\right\|}
\newcommand{\twonorm}[1]{\left\|#1\right\|_2}

\newcommand{\decnorm}[2]{\left\|#1\right\|_{#2}}
\newcommand{\fronorm}[1]{\decnorm{#1}{F}}

\newcommand{\ip}[2]{\ensuremath{\left\langle #1,#2\right\rangle}}

\newcommand{\inset}[1]{\left\{#1\right\}}

\newcommand{\inparen}[1]{\left(#1\right)}

\newcommand{\ind}[1]{\ensuremath{\mathbf{1}_{#1}}}

\newcommand{\eps}{\varepsilon}
\renewcommand{\epsilon}{\varepsilon}

\newcommand{\mkw}[1]{{\footnotesize  [\textbf{\textcolor{blue}{#1}} \textcolor{blue!60!black}{--mary}]\normalsize}}
\newcommand{\TODO}[1]{\textcolor{red}{\textbf{TODO: #1}}}

\newtheorem{theorem}{Theorem} 
 
\newtheorem{definition}{Definition}

\newtheorem{corollary}[theorem]{Corollary}

%% file: intro.tex
We consider a distributed setting where a master wants to run a gradient-descent-like algorithm to solve an optimization problem distributed across several workers. Let $X \in \mathbb{R}^{m \times \ell}$ be a data matrix and let $\mathbf{x}_i \in \mathbb{R}^\ell$ denote the $i$'th row of $X$. 
Let $\mathbf{y} \in \mathbb{R}^m$ be a vector of labels, so $\bx_i$ has label $y_i$.
Define $A \triangleq [X|\mathbf{y}]$ to be the concatenation of $X$ and $\mathbf{y}$.
The master wants to find a vector $\bbeta^* \in \mathbb{R}^\ell$ that best represents the data $X$ as a function of the labels $\mathbf{y}$.
That is, the goal is to iteratively solve an optimization problem
\begin{equation}
\label{eq:optprob}
 \bbeta^* = \arg\min_{\bbeta} \mathcal{L}(A, \bbeta),
\end{equation}
for a given loss function $\mathcal{L}$, by simulating or approximating an update rule of the form
\begin{equation}\label{eq:gdupdate}
 \bbeta_{t+1} = \bbeta_t - \gamma_t \nabla \mathcal{L}(A, \bbeta_t). 
\end{equation}
Many natural loss functions 
$\mathcal{L}(A, \bbeta)$ can be written as the sum over individual rows 
$\mathbf{a}_i$ of $A$, i.e., 
\begin{equation}
\label{eq:gradient}
 \mathcal{L}(A, \bbeta) = \sum_{i=1}^m \mathcal{L}(\mathbf{a}_i, \bbeta),
\end{equation}
such loss functions lend themselves naturally to distributed algorithms.
In a distributed setting, 
the master partitions the data matrix $A$ into rows $\mathbf{a}_i$ which are distributed between the workers. Each worker returns some linear combination(s) of the gradients $\nabla \mathcal{L}(\mathbf{a}_i, \bbeta_t)$ that it can compute, and the master aggregates these together to compute or approximate the update step~\eqref{eq:gdupdate}.

We focus on the setting where some of the workers may be \em stragglers, \em i.e., slow or unresponsive.  This setting has been studied before in the systems community  \cite{DB13,chen2016revisiting,ananthanarayanan2010reining,zaharia2008improving}, and recently in the coding theory community \cite{tandon2017gradient, ye2018communication, lee2018speeding}. A typical approach is to introduce some redundancy: for example, the same piece of data $\mathbf{a}_i$ might be held by several workers. There are several things that one might care about in such a scheme and in this paper we focus on the following four desiderata:
\begin{itemize}
\item[(A)] \textbf{Convergence speed.} We would like the error $\|\bbeta_t - \bbeta^*\|_2$ to shrink as quickly as possible. 
\item[(B)] \textbf{Redundancy.} We would like to minimize the amount of storage and computation overhead needed between the workers.
\item[(C)] \textbf{Communication.} We would like to minimize the amount of communication between the master and the workers.
\item[(D)] \textbf{Flexibility.} In practice, there is a great deal of variability in the number of stragglers over time.  We would like an algorithm that degrades gracefully if more stragglers than expected occur. 
\end{itemize}


Much existing work has focused on simulating gradient descent \em exactly, \em even in the presence of worst-case stragglers, for example \cite{tandon2017gradient, ye2018communication, raviv2017gradient, lee2018speeding}. 
In that model, at each round an arbitrary set of $s$ workers (for a fixed $s$) may not respond to the master. The goal is for the master to obtain the same update $\bbeta_t$ at round $t$ that gradient descent would obtain. For this to happen, the master should be able to obtain an exact value of the gradient $\nabla\mathcal{L}(A, \bbeta_t)$.
This has given rise to (exact) \em gradient coding \em \cite{tandon2017gradient}, which focuses on optimizing desiderata (A) and (C) above.  However, these schemes (and necessarily, any scheme in this model) do not do so well on (B) and (D).  First, it is not hard to see that in the presence of $s$ worst-case stragglers, it is necessary for any $n-s$ workers to be able to recover all of the data, which necessitates a certain amount of overhead. Namely, every data vector should be replicated on $s+1$ different workers.  Second, the gradient coding schemes for example in \cite{tandon2017gradient,ye2018communication} are brittle in the sense that they work perfectly for $s$ failures, but cannot handle more than $s$ stragglers.  

On the other hand, there has also been work on \em approximately \em simulating gradient descent.  
A first approach in this direction (similar in spirit to the method in \cite{chen2016revisiting}) 
is to assume that the stragglers are random, rather than worst-case, and not employ any redundancy at all. Thus, the master obtains an approximate update \eqref{eq:gdupdate} instead of an exact one by computing the sum in \eqref{eq:gradient} without the responses of the stragglers. (We will later refer to this algorithm as ``Ignore--Stragglers--SGD.'')
If the stragglers are independent at each round, this algorithm is a close approximation to Batch--SGD, see e.g. \cite{bottou1998online,shalev2011stochastic,gimpel2010distributed,shalev2011pegasos}, and performs in about the same way. 
However, for convex loss functions it is well known that, while Batch--SGD does converge to $\bbeta^*$, the convergence is not as fast as that of classical gradient descent \cite{SGDopt,bottou2018optimization,boyd2004convex}.  Thus, this approach maintains the good communication cost (C) of the coded approaches by requiring each worker to send one linear combination of the gradients to the master, and improves on (B) and (D), but sacrifices (A), the convergence rate.

A line of work known as \em approximate gradient coding \em \cite{charles2017approximate, maity2018robust,raviv2017gradient, wang2019erasurehead,wang2019fundamental,horii2019distributed} introduces redundancy in order to speed up the convergence rate of such an approximate scheme.
This line of work studies the data redundancy $d$ (that is, the number of times each row $\mathbf{a}_i$ of the data matrix $A$ is replicated) needed to tolerate $s$ stragglers and allow the master to compute an approximation of the gradient if more than $s$ workers are stragglers \cite{charles2017approximate,raviv2017gradient, wang2019erasurehead,wang2019fundamental}. In \cite{maity2018robust} a variant of this idea is studied; in that work the data is encoded using LDPC code rather than being duplicated. 
In approximate gradient coding, the master is required to compute the exact gradient with high probability if fewer than $s$ workers are stragglers. If more than $s$ workers are stragglers, the distance between the computed gradient at the master and the true gradient can be made small if the redundancy factor is poly-logarithmic in the number of workers. So far, this line of work has mostly focused on desiderata (B), (C) and (D), and most works have not directly analyzed the convergence time (A).
Two exceptions are \cite{maity2018robust} and \cite{wang2019erasurehead}, which we discuss more below.

In this work, we introduce an approximate gradient coding scheme called \em Stochastic Gradient Coding \em (SGC) which works in the random straggler model and which does well simultaneously on desiderata (A)-(D).
We analyze the convergence rate of SGC, and we 
present experimental work which demonstrates that SGC 
outperforms the most recently proposed schemes \cite{charles2017approximate,wang2019erasurehead} when $p$ (the fraction of workers that the master will ignore in each iteration) is relatively large.\footnote{We note that rather than thinking of workers as being unresponsive with probability $p$, we are motivated by a setting where the master waits for the $1-p$ fastest fraction of the workers before proceeding to the next iteration.  This setting motivates the case where $p$ is relatively large, which is the setting we focus on in this work.}

\subsection{Contributions} 

We consider an approach that we call \em Stochastic Gradient Coding \em (SGC). The idea---which is similar to previous approaches in approximate gradient coding~\cite{charles2017approximate}---is simple: the master distributes data to the workers with a small amount of repetition according to a \em pair-wise balanced scheme \em (which we will define below); a data point $\mathbf{x}_i$ is replicated $d_i$ times, and $d_i$ can vary from data point to data point.  Below, the redundancy parameter $d$ refers to the average of the $d_i$'s.  
Once the data is distributed,
the algorithm proceeds similarly to the Ignore--Stragglers--SGD algorithm described above: workers compute gradients on their data and return a linear combination, and the master aggregates all of the linear combinations it receives to do an update step.  

One contribution of this work is to provide a rigorous convergence analysis of SGC.  We show that
SGC with only a small amount of redundancy $d$ is able to regain the benefit of (A) from the (exact) coded approaches, while still preserving the benefits of (B), (C), (D) that the ``Ignore--Stragglers--SGD'' approach sketched above does.  
A second contribution is extensive experimental evidence which suggests that for the same small redundancy factor $d$ SGC outperforms other schemes when there are many stragglers.

More precisely, our contributions are as follows (all in the stochastic straggler model):
\begin{itemize}
	\item In the special case of the $\ell_2$ loss function, we show that SGC with redundancy factor $d>1$, can obtain error bounds where $\|\bbeta^* - \bbeta_t\|_2$ decreases at first exponentially and then proportionally to $\frac{1}{td}$.  This mirrors existing results on SGD (which corresponds to the case $d=1$), and quantifies the trade-off between replication and error.  This is made formal in Theorem~\ref{thm:convl2}.
	\item For more general loss functions, we show that SGC has at least the same convergence rate as Ignore--Stragglers--SGD, and we give some theoretical evidence that the error $\|\bbeta^* - \bbeta_t\|$ may decrease as $d$ increases.  This is made formal in Theorem~\ref{thm:conv}. 
	\item We provide numerical simulations comparing SGC to gradient descent, Ignore--Stragglers--SGD and a few other versions of SGD, and other approximate gradient coding methods.  Our simulations show that indeed SGC improves the accuracy of Ignore--Stragglers--SGD, with far less redundancy than would be required to implement exact gradient descent using coding. In addition, we compare SGC to other approximate gradient methods existing in the literature and show that SGC outperforms the existing methods when the probability of workers being stragglers is high.
	\end{itemize}

\subsection{Relationship to previous work on approximate gradient coding}

We provide a more detailed description of previous work in Section~\ref{sec:related}, but first we briefly mention some of the main differences between our work and existing work on approximate gradient coding \cite{charles2017approximate, maity2018robust,raviv2017gradient, wang2019erasurehead,wang2019fundamental,horii2019distributed}.

First, we note that our SGC scheme is quite similar to Bernoulli Gradient Coding (BGC) studied in \cite{charles2017approximate}, where the data is distributed uniformly at random to $d$ workers.  One difference between our work and that work is that we allow for the redundancy of different data points $\mathbf{x}_i$ to vary for different $i$; we will see that for the $\ell_2$ loss function it makes sense to choose $d_i$ based on $\|\mathbf{x}_i\|_2.$  A second difference between our work and \cite{charles2017approximate} is that \cite{charles2017approximate} does not provide a complete convergence analysis.  {The works \cite{raviv2017gradient,wang2019fundamental,horii2019distributed} also study schemes similar in flavor to SGC, but these works also do not provide  complete convergence analyses.}

%

The works of \cite{maity2018robust, wang2019erasurehead} do provide convergence analyses, although for schemes that are quite different from SGC.  More precisely, \cite{maity2018robust} studies a scheme with LDPC coding, rather than repetition.  The work of \cite{wang2019erasurehead} studies a scheme based on Fractional Repetition (FR) codes, which was proposed in \cite{charles2017approximate}.  This scheme partitions the data and the workers into different blocks; every worker in a block receives all of the data from the corresponding block.

Additionally, we obtain slightly different error guarantees than the analyses of \cite{maity2018robust, wang2019erasurehead}.  More precisely, the analysis of \cite{wang2019erasurehead} proves a bound where the error decreases exponentially in $T$ (the number of iterations of the algorithm) until some noise floor is hit.  The analysis of \cite{maity2018robust} studies the special case of the $\ell_2$ loss function, and shows that the error decays like $\mathcal{O}(1/\sqrt{T})$.  In contrast, for SGC and for the special case of the $\ell_2$ loss function, we show that the error decays exponentially in $T$ at first and then switches to dacaying like $\mathcal{O}(1/T)$; this mirrors existing results for SGD for the $\ell_2$ loss function.  We give a more general result that holds for general convex loss functions and show that the error decays as $\mathcal{O}(1/T)$.  

Finally, we provide empirical results which suggest that our scheme can outperform existing gradient coding schemes (in particular, the FR-based approach of \cite{charles2017approximate,wang2019erasurehead} and BGC~\cite{charles2017approximate}) in some parameter regimes.  We do not compare our scheme empirically to that of \cite{maity2018robust,horii2019distributed} because they requires more work on the master's end (to encode and decode) and are thus not directly comparable to our work.

\subsection{Organization}
We give a more precise definition of our set-up in Section~\ref{sec:model}.
We describe the SGC algorithm in Section~\ref{sec:SGC}.
In Section~\ref{sec:results}, we give a more detailed overview of both our theoretical and empirical results, which are fleshed out in Sections~\ref{sec:convergence} and \ref{sec:sims} respectively.  The proofs of our results can be found in Section~\ref{sec:proofs}.  We provide more detail on related work in Section~\ref{sec:related}.

%% file: model.tex
\subsection{Probabilistic model of stragglers}
In this paper, we adopt a probabilistic model of stragglers.  More precisely, we assume that at every iteration each worker may be a straggler with some probability $p$, and this is independent between workers and between iterations.  This is a strong assumption, but it is a natural starting place.%
\footnote{In our numerical simulations, we relax the assumption of independence and show that similar results hold when the identities of the stragglers are somewhat persistent from round to round and change only after a fixed number of iterations.}
Our probabilistic model is similar to the model in  \cite{charles2017approximate, raviv2017gradient, wang2019erasurehead, wang2019fundamental,maity2018robust, horii2019distributed} and is in contrast to the worst-case model assumed by much of the literature on coded computation. 

\subsection{Computational model}
Our computational model has two stages, a  distribution stage and a computation stage.

In the \em  distribution stage, \em the master decides which data to send to each worker.  More precisely, the master can decide to send each row $\mathbf{a}_i$ of $A$ to $d_i$ different workers.  We refer to the parameter $d = \frac{1}{m} \sum_{i=1}^m d_i$ as the \em redundancy \em  of the scheme.

The \em computation stage \em is made up of rounds, each of which contains two repeating steps.  In the first step, the master does some local computation and then sends a message to each worker.  In the second step, each worker does some local computation and tries to send a message back to the master; however, with probability $p$ the message may not reach the master.  Then the round is over and the master repeats the first step to begin the next round.  We refer to the total amount of communication per round as the \em communication cost \em of the scheme.

%% file: SGC.tex
In this section, we describe our solution, which we call \em Stochastic Gradient Coding \em (SGC).  The idea behind SGC is extremely simple.  It is very much like the Ignore--Stragglers--SGD algorithm described above, except we introduce a small amount of redundancy.  We describe the distribution stage and the computation stage of our algorithm below.  Our scheme has parameters $d_1, \ldots, d_m$, which control the redundancy of each row, and a parameter $\gamma_t$ which controls the step size.  We will see in the theoretical and numerical analyses how to set these parameters.

In our analysis, we focus on \em pair-wise balanced schemes: \em
\begin{definition}
We say that a distribution scheme that sends $\mathbf{a}_i$ to $d_i$ different workers is \em pair-wise balanced \em if for all $i \neq i'$, the number of workers that receives $\mathbf{a}_i$ and $\mathbf{a}_{i'}$ is $\frac{d_i d_{i'}}{n}$.
\end{definition}
Notice that with a completely random distribution scheme, the expected number of workers who receive both $\mathbf{a}_i$ and $\mathbf{a}_{i'}$ for $i \neq i'$ is equal to $\frac{ d_i d_{i'} } {n}$.
In our analysis, it is convenient to deal with schemes that are exactly pair-wise balanced.  However, for small $d_i$ it is clear that no such schemes exist (indeed, we may have $\frac{d_i d_{i'}}{n} < 1$).  In our simulations, we choose a uniformly random scheme\footnote{In our simulations, we assign rows to $d_i$ workers uniformly at random, which approximates a pair-wise balanced scheme. Similarly, the BGC construction of \cite{charles2017approximate} approximates a pair-wise balanced scheme where each row is assigned to $d$ workers uniformly at random, i.e., $d_i=d$ for all $i\in [m]$.} which seems to work well (see Section~\ref{sec:sims}).  We believe that our analysis should extend to a random assignment as well, although for simplicity we focus on pair-wise balanced schemes in our theoretical results.

The way SGC works is as follows:
\begin{itemize}
\item \textbf{Distribution Stage}.
The master creates $d_i$ copies of each row $\mathbf{a}_i$, $i=1,\dots,m$, and sends them to $d_i$ distinct workers according to a pair-wise balanced scheme.
We denote by $S_j$, $j=1,\dots,n$, the set of indices of the data vectors given to worker $\w_j$, i.e., $S_j = \{i; \mathbf{a}_i \text{ is given to } \w_j\}$.
\item \textbf{Computation Stage.}
At each iteration $t$, the master sends $\mathbf{\bbeta}_t$ to all the workers. Each worker $\w_j$ computes 
\begin{equation}\label{eq:sumw}
f_j (\bbeta_t) \triangleq \gamma_t \sum_{i\in S_j} \dfrac{1}{d_i(1-p)}\nabla \mathcal{L}(\mathbf{a}_i, \mathbf{\bbeta}_t)
\end{equation}
and sends the result to the master. The master aggregates all the received answers from non straggler workers, sums them and updates $\bbeta$ as follows:
\bes
\bbeta_{t+1} = \bbeta_t - \gamma_t\sum_{j=1}^n \sum_{i=1}^m \dfrac{\mathcal{I}_i^j}{d_i(1-p)} \grad{i}{t},
\ees
where $\mathcal{I}_i^j$  is the indicator function for worker $j$ being non straggler and having obtained point $\mathbf{a}_i$ during the data distribution, i.e., 
\bes
\mathcal{I}^j_i = 
\begin{cases}
1 & \hfill \text{if worker $j$ is non straggler and has point }\mathbf{a}_i,\\
0 & \hfill \text{otherwise}.
\end{cases}
\ees
Note that  $\mathcal{I}_i^j$    depends on the iteration $t$, however we drop $t$ from the notation for notational convenience  since the value of  $t$ will be clear from the context.
\end{itemize}
For use below, we define 
\begin{equation}\label{eq:ghat}
\hat{\mathbf{g}}_t\triangleq\sum_{j=1}^n \sum_{i=1}^m \dfrac{\mathcal{I}_i^j}{d_i(1-p)} \grad{i}{t}.
\end{equation}
\noindent We call $\hat{\mathbf{g}}_t$ the estimate of the gradient at iteration $t$ which estimates the exact gradient of the loss function in \eqref{eq:gradient}, $$\mathbf{g}_t \triangleq \sum_{i=1}^m \grad{i}{t}.$$

%% file: results_summary.tex
In this section, we summarize both our theoretical and numerical results.

\subsection{Theoretical results}
Our main theoretical contributions are to derive results for SGC that mirror known results for SGD and Batch--SGD.
There are two important differences between our results and those for Batch--SGD.
\begin{enumerate}
\item First, one of our goals is to show how the error $\|\bbeta^* - \bbeta_t\|_2^2$ depends on the redundancy parameter $d$; we show that it is roughly like $1/d$.  This explains why SGC can work much better than Ignore--Stragglers--SGD (say, so that $\|\bbeta_t - \bbeta^*\|_2^2$ is half as large), even with relatively low redundancy (say, $d=2$).  In Batch--SGD we always have $d=1$. 
\item Second, it is nontrivial to adapt existing results for Batch--SGD to our setting.  The reason is that the batches are not uniform in our setting; rather, they depend on the way that the data is distributed.  We note that this is true even if the data is distributed randomly to begin with: in that case it is true that the marginals of the batches are uniformly random (that is, in each round the set of gradients that the master receives is a uniformly random subset of all of them) but because the randomness from the initial distribution is fixed throughout the computation, if we view it this way then the batches are no longer independent.  The main technical challenge in our analysis (in particular, the proof of Theorem~\ref{thm:informalsgcl2} below) is to deal with this issue.\footnote{We note that this is not an issue for our proof of Theorem~\ref{thm:informalsgc}, since we are able to adapt existing results that depend only on the mean and variance of the gradient estimates.}
\end{enumerate}

We adapt existing result from the SGD literature to prove a tighter bound that holds for arbitrary convex loss functions. And we derive a stronger convergence guarantee for the $\ell_2$ loss function.

\textbf{Special case: $\ell_2$ loss function.}
We begin with a result which is specialized for the $\ell_2$ loss function.  This result
is of a similar flavor as the results of \cite{SV09,NWS14,NW16} on SGD and the randomized Kaczmarz algorithm.\footnote{We note that \cite{NWS14} also holds for more general loss functions.}
Those works show that the speed of convergence is exponential to begin with, and then begins to decay polynomially like $1/t$ once an unavoidable limit is reached. 
In this work, we show an analogous result for the $\ell_2$ loss function.
in this case we show that the convergence is exponential to begin with, until the noise is on the order of $\ell_2$ normal of the residual $\mathbf{r} \triangleq X\bbeta^* - \mathbf{y}$, and then it begins to decay polynomially like $1/(dt)$. 

Thus, our analysis generalizes the case when $d=1$ (aka, Ignore--Stragglers--SGD), and we see that as the repetition factor $d$ increases, the error of SGC decreases.  We state our main theorem informally below, and we state the formal version in Section~\ref{sec:convergence}. Throughout the paper we abuse notation use the superscript $T$ to denote the transpose of a matrix.

\begin{theorem}[Informal; see Theorem~\ref{thm:convl2} for a formal version]\label{thm:informalsgcl2}
Consider an SGC algorithm run on a matrix $A\triangleq[X|\mathbf{y}]$ of dimension $m \times (\ell+1)$ distributed to $n$ workers.
Suppose that the distribution scheme is pairwise balanced, 
and that each row $\ba_i$ of $A = [X|\mathbf{y}]$ is sent to $d_i$ different workers, where $d_i$ is chosen proportional to $\|\bx_i\|_2^2$. 

Suppose than $n$ is sufficiently large and that 
\[ d = \frac{1}{m} \sum_{i=1}^m d_i \geq 8 \inparen{ \frac{p}{1-p} }. \]

Choose an error tolerance $\eps > 0$.
Then, it is possible to choose a step size $\gamma_t$ at each step $t$ so that the following guarantee holds on the iterates $\bbeta_T$ of SGC, for $T \geq 2\log(1/\eps^2)$:

\[ \EE \left[ \|\mathbf{\bbeta}_T - \mathbf{\bbeta}^*\|^2_2 \right] \leq \eps^2 \| \mathbf{\bbeta}_0 - \mathbf{\bbeta}^*\|^2_2 + \frac{1}{d\cdot T} \cdot \inparen{\log^2(1/\eps)\frac{p}{1-p}} \cdot  \|\tilde{\mathbf{r}}\|^2, \]
where $\tilde{\mathbf{r}} = (X\mathbf{\bbeta}^* - \mathbf{y})/ \|X^TX\|_2$.
\end{theorem}

That is, if the residual $\tilde{\mathbf{r}}$ is very tiny, so that the second term is smaller than the first, then the algorithm reaches accuracy $\eps$ in roughly $\log(1/\eps)$ steps.  However, if $\tilde{\mathbf{r}}$ is larger, then the convergence becomes polynomial, matching what we expect from SGD.  In this second case, the difference is that the replication factor $d$ appears in the denominator, so that when $d$ is larger, the error is smaller, explaining why replication helps.
Notice that if $p$ is constant, we expect good performance when $d = O(1)$.  In contrast, to exactly simulate gradient descent via coding would require $d = \Omega(n)$.

The main difficulty in proving Theorem~\ref{thm:convl2} (the formal version of Theorem~\ref{thm:informalsgcl2})
is that because the data distribution is fixed ahead of time, the ``batches'' that the master acquires in each round are not uniformly random, but rather come from some distribution determined by the data distribution.

\textbf{Beyond $\ell_2$ loss function.}
Our result above is limited in that it only applies to the $\ell_2$ loss function.  We believe that the analysis of Theorem~\ref{thm:convl2} should apply to general loss functions, but for now we observe that in fact a convergence rate of $1/t$ does follow for SGC from a result of \cite{SGDopt}.

In that work, the authors give a general analysis of stochastic gradient descent, which works as long as (in our language) the master is computing an unbiased estimator of the gradient.  The convergence speed of the algorithm then depends on the variance of this estimate.   This result applies in our setting:

\begin{theorem}[Informal; see Theorem~\ref{thm:conv} for a formal version]\label{thm:informalsgc}
Suppose that SGC is run on a matrix $A\triangleq [X | \mathbf{y}]$ of dimension $m \times (\ell+1)$ distributed to $n$ workers. Suppose that the distribution scheme is pairwise balanced, and that each row $\ba_i$ of $A$ is sent to $d_i$ different workers, $d_i\leq n$. Consider a version of the optimization problem in \eqref{eq:optprob} where $\bbeta$ is constrained to a convex set $\mathcal{W}$.  Under some mild assumptions on the loss function $\mathcal{L}$ and assuming there exists a constant $C$ such that
\[ \norm{\nabla \mathcal{L} (\mathbf{a}_i,\bbeta)}_2^2 \leq C^2 \]
for all $i \in [n]$ and for all $\bbeta \in \mathcal{W}$, then 
there is a way to choose the step size $\gamma_t$ at each step $t$ so that 
the error after $T$ iterations is bounded by 
\[ \mathbb{E} \| \bbeta_T - \bbeta^* \|_2^2 \leq \mathcal{O}(1/T). \]
\end{theorem}


The proof of Theorem~\ref{thm:informalsgc} (given Lemma 1 in \cite{SGDopt}) boils down to showing that our gradient estimator $\hat{\mathbf{g}}_t$ is an unbiased estimator of the true gradient and that $\EE \| \hat{ \mathbf{g} }_t \|_2^2$ is bounded for all $t$, 
which we do in Section~\ref{sec:proofs}.

We give more precise statements of these theorems in Section~\ref{sec:convergence}, and prove them in Section~\ref{sec:proofs}.

\subsection{Numerical simulations}
\begin{figure}
\setlength\figureheight{0.5\textwidth}
\setlength\figurewidth{0.7\textwidth}
\centering
\begin{minipage}{0.45\textwidth}
\resizebox{0.95\textwidth}{!}{
\input{intropb=0.5-w_-n=10,d=2,noise=1,datapoints=500,dim=50,iter=500str=0}
}
\caption{Comparison between Ignore--Stragglers--SGD, SGC and {{\erasurehead}} in terms of the distance between $\bbeta_t$ and $\bbeta^*$ the value of $\bbeta$ that minimizes the loss function. SGC outperforms Ignore--Straggler--SGD at the expense of adding small redundancy, $d=2$ in this example.}
\label{fig:intro1}
\end{minipage}%
\hfill%
\begin{minipage}{0.45\textwidth}
\resizebox{0.95\textwidth}{!}{
\input{introconvergence-n=10,d=2,noise=1,datapoints=500,dim=50,iter=500pow=0.7str=0.tex}
}
 \caption{Final convergence of all algorithms run for $T=5000$ iterations as function of $p$ the probability of workers being stragglers. We omit GD in this setting, because it has the same performance as all algorithms when $p=0$.}
 \label{fig:intro2}
\end{minipage}

\end{figure}

We run extensive simulations on synthetic data $A$ of dimension $1000\times 100$ generated from a Gaussian distribution. 
We compare SGC to four other algorithms detailed in Section~\ref{sec:sims} and show that SGC outperforms all other algorithms when there are many stragglers. 
A typical result is shown in Figure~\ref{fig:intro1}.  In it, we observe that SGC and {{\erasurehead}} outperform Ignore--Stragglers--SGD at the expense of doubling the redundancy. In Figure~\ref{fig:intro2} we plot the convergence of approximate gradient codes as function of $p$. We observe that SGC outperforms {{\erasurehead}} when the number of stragglers is large, $p>0.6$. As expected, the approximate algorithms have worse accuracy than full-blown gradient descent, but we note that implementing exact gradient descent with a $p$ fraction of stragglers would require redundancy $d \approx pn \gg 2$. Moreover, we observe the flexibility of the approximate algorithms in the number of stragglers, and we note that computing GD exactly would lack this flexibility. In Section~\ref{sec:sims}, we comment on how the dependency between stragglers affect the convergence of SGC. Our implementation is publicly available \cite{SGCGit}.


%% file: introconvergence-n=10,d=2,noise=1,datapoints=500,dim=50,iter=500pow=0.7str=0.tex
\begin{tikzpicture}

\definecolor{color0}{rgb}{0.12156862745098,0.466666666666667,0.705882352941177}
\definecolor{color1}{rgb}{1,0.498039215686275,0.0549019607843137}
\definecolor{color2}{rgb}{0.172549019607843,0.627450980392157,0.172549019607843}
\definecolor{color3}{rgb}{0.83921568627451,0.152941176470588,0.156862745098039}
\definecolor{color4}{rgb}{0.580392156862745,0.403921568627451,0.741176470588235}
\definecolor{color5}{rgb}{0.549019607843137,0.337254901960784,0.294117647058824}
\definecolor{color6}{rgb}{0.890196078431372,0.466666666666667,0.76078431372549}

\begin{axis}[
title={$n = 10$ workers},
title style = {font=\Large},
xlabel={Probability of straggling $p$},
ylabel={$\norm{\bbeta_T - \bbeta^*}$},
ylabel style = {font = \Large, at={(-0.04,0.5)}},
xlabel style = {font=\Large},
xmin=-0.045, xmax=0.945,
ymin=1.68644950884144e-08, ymax=0.80029285505015,
ymode=log,
width=\figurewidth,
height=\figureheight,
tick align=outside,
tick pos=left,
x grid style={white!69.01960784313725!black},
y grid style={white!69.01960784313725!black},
legend style={at={(0.03,0.97)}, anchor=north west, draw=white!80.0!black},
legend cell align={left}
]

\def\msize{2}

\addplot [ultra thick, color1, dashed, mark=*, mark size=\msize, mark options={solid}]
table {%
0 3.76612708530429e-08
0.1 3.64333587466363e-05
0.2 5.25404994195889e-05
0.3 7.7615181021457e-05
0.4 8.82456978241633e-05
0.5 9.3760156052778e-05
0.6 0.000124319114537291
0.7 0.000145097726659664
0.8 0.000211614411280793
0.9 0.000286440597633678
};
\addlegendentry{Ignore--Stragglers--SGD}

\addplot [ultra thick, color0, mark=square*, mark size=\msize, mark options={solid}]
table {%
0 3.76612708530429e-08
0.1 2.50850184492928e-05
0.2 3.56732644460501e-05
0.3 5.24751138616398e-05
0.4 6.09880379720325e-05
0.5 5.89368728831169e-05
0.6 8.56997507201122e-05
0.7 9.82836711012897e-05
0.8 0.000140529513791184
0.9 0.000193389169903699
};
\addlegendentry{SGC}

\addplot [ultra thick, color3, mark = +, mark size = \msize, mark options = {solid}]
table {%
0 3.76612708530429e-08
0.1 9.64861116289711e-06
0.2 1.93644890211481e-05
0.3 2.69691399160198e-05
0.4 4.14381715419961e-05
0.5 4.57826355068105e-05
0.6 7.431347568075e-05
0.7 0.000394148935463102
0.8 0.00466606686584583
0.9 0.0455991457852139
};
\addlegendentry{\sc ErasureHead}
\end{axis}

\end{tikzpicture}

%% file: convergence.tex
In this section we precisely state our theoretical results.   We begin with a specialized result for the $\ell_2$ loss function, and then include a result for more general loss functions.

\subsection{Special case: $\ell_2$ loss function}
\label{sec:5A}
We begin with a result that holds for the special case of an $\ell_2$ loss function, aka, regression.  
Inspired by the approach of \cite{NW16} for SGD, our approach is to consider a \em weighted \em distribution scheme; that is, we choose $d_i$ proportionally to $\| \mathbf{x}_i\|_2^2$.  While the statement below is only for the $\ell_2$ loss function, we conjecture that it holds for more general loss functions. 

Define a parameter 
\[ \mu = \frac{ \frac{1}{m} \fronorm{X}^2 }{ \| X^T X \| }. \]
This parameter measures how incoherent $X$ is.  If $X$ is orthogonal, $\mu = 1$, while if, for example, $X$ is the all-ones matrix, then $\mu = 1/m$.  It is not hard to check that $\mu \in [0, 1]$.  

Suppose that $\mathcal{D}$ is a pair-wise balanced distribution scheme which sends $\mathbf{a}_i$ to $d_i$ different workers, where
\begin{align}\label{eq:di}
d_i &=  \sigma \cdot \|x_i\|_2^2,\\
 \sigma &= \frac{nd}{\fronorm{X}^2} = \frac{d}{\mu \|X^T X\|} \label{eq:sigma},\\
 d &= \frac{1}{m} \sum_{i \in [m]} d_i. \label{eq:d}
 \end{align}

Notice that, as stated, it is possible that the $d_i$ end up being non-integral; in the following, we will assume for simplicity below that $d_i \in \mathbb{Z}$ for all $i$.  Notice that if $\|x_i\|_2 = 1$ for all $i$, then this will be the case because we can choose $d_i = d$ to be any integer we choose, and this defines $\sigma$.

\begin{theorem}\label{thm:convl2}
Consider an SGC algorithm run on a matrix $A\triangleq[X|\mathbf{y}]$ of dimension $m \times (\ell+1)$ distributed to $n$ workers according to a pairwise balanced distribution scheme with $d_i$ as described above, with loss function
\[ \mathcal{L}([X|\mathbf{y}], \bbeta) = \twonorm{ X\bbeta - \mathbf{y} }^2, \]
and assume that the degrees $d_i \leq n$ are all integral. 

Suppose the stragglers follow the stochastic model of Section~\ref{sec:model}, and that each worker is a straggler independently with probability $p$.  Choose $\eps > 0$ and choose $T \geq 2 \log(1/\eps^2)$.

Suppose that the number of workers $n$ satisfies
$n \geq 8\inparen{ \frac{p}{1 - p} }$, and that 
\[  8 \mu \inparen{ \frac{p}{1 -p} } \leq d . \] 
Choose a step size
\[\gamma_t =  \frac{1}{\norm{ X^T X }} \cdot \min \inset{ \frac{1}{2} , \dfrac{ \log(1/\eps^2 ) }{ t } }. \]
Then, after $T$ iterations of SGC, we have 
\[ \EE \left[ \|\mathbf{\bbeta}_T - \mathbf{\bbeta}^*\|^2_2 \right] \leq \eps^2 \| \mathbf{\bbeta}_0 - \mathbf{\bbeta}^*\|^2_2 + \frac{1}{dT} \inparen{ \log^2(1/\eps^2) \inparen{\frac{p}{1-p}} \|\tilde{\mathbf{r}}\|^2 \mu} \]
where the expectation is over the stragglers in each of the $T$ iterations of SGC and where $\mu$ is as above, and where
\[ \tilde{\mathbf{r}} = \dfrac{\|X\mathbf{\bbeta}^* - \by\|_2^2}{\|X^T X\|_2}. \]

\end{theorem}

\begin{corollary}
Suppose that $X \mathbf{\bbeta}^* = \mathbf{y}$ (that is, we are solving a system for which there is a solution) and that $n \geq 8p/(1-p)$. Then the algorithm described in Theorem~\ref{thm:convl2} converges with
\[ \EE \left[\|\mathbf{\bbeta}_T - \mathbf{\bbeta}^* \|_2^2\right] \leq \eps^2 \|\mathbf{\bbeta}_0 - \mathbf{\bbeta}^*\|^2 \]
provided that $T \geq 2 \log(1/\eps^2)$ and $d \geq 8\mu p/(1 - p)$.
\end{corollary}
In particular, since $\mu \leq 1$, this says that we need to take $d \gtrsim p/(1-p)$ and the algorithm converges extremely quickly. 


\subsection{Beyond $\ell_2$ loss function}
Now, we consider a constrained version of the problem given in~\eqref{eq:optprob}, where $\bbeta$ belongs to a bounded set $\mathcal{W}$. In this section, we state a result for general loss functions $\mathcal{L}$ which are $\lambda$-strongly convex:
\begin{definition}[Strongly convex function]
A function $\mathcal{L}$ is $\lambda$-strongly convex, if for all $\bbeta$, $\bbeta'$ $\in \mathbb{R}^{\ell}$ and any subgradient $\mathbf{g}$ of $\mathcal{L}$ at $\bbeta$,
\begin{equation}
\mathcal{L}(\bbeta')\geq \mathcal{L}(\bbeta)+\ip{ \mathbf{g} }{\bbeta'-\bbeta} + \dfrac{\lambda}{2} \norm{\bbeta'-\bbeta}_2^2.
\end{equation}
\end{definition}

%
%


Theorem~\ref{thm:conv} below follows from the analysis in \cite{SGDopt}.

\begin{theorem}\label{thm:conv}
Suppose that SGC is run
on a matrix $A\triangleq [X | \mathbf{y}]$ of dimension $m \times (\ell+1)$ distributed to $n$ workers with each row of $A$ sent to $d_i$ different workers, $d_i\leq n$, according to a pairwise balanced distribution scheme. Consider a version of the optimization problem in \eqref{eq:optprob} where $\bbeta$ is constrained to a convex set $\mathcal{W}$, i.e., $$\bbeta^* = \displaystyle \arg\min_{\bbeta \in \mathcal{W}} \mathcal{L}(A,\bbeta),$$ and at each step of the algorithm
$\bbeta_{t+1} = \Pi_{\mathcal{W}}(\bbeta_t - \gamma_t \ghatt),$ where $\Pi$ is the projection operator.
Let $p$ denote the probability of a given worker being a straggler at a given iteration. 
Suppose that the 
loss function $\mathcal{L}$ is $\lambda$-strongly convex with respect to the optimal point $\bbeta^* \in \mathcal{W}$, and that all of the partial gradients $\nabla \mathcal{L}( \ba_i, \bbeta )$ are bounded for $i \in [n]$ and $\bbeta \in \mathcal{W}$, i.e. there exists a constant $C$ so that 
$$\norm{\nabla \mathcal{L} (\mathbf{a}_i,\bbeta)}_2^2 \leq C, \qquad \forall \bbeta \in \mathcal{W}, i \in [n].$$ 
Suppose that the step size is set to be $\gamma_t = 1/(\lambda t)$.  Then the error after $T$ iterations is bounded by 
\begin{equation}\label{eq:ineq1}
\mathbb{E}\| \bbeta_T - \bbeta^* \|_2^2%
\leq \dfrac{4}{\lambda^2 T}\cdot  mC^2 \inparen{ \frac{p}{1-p} \cdot \frac{1}{d_{\min} } +  \frac{ (m-1) p }{n (1-p) } + m } 
\end{equation}
where $d_\text{min} \triangleq \displaystyle \min_{i\in [m]} d_i$. 

\end{theorem}
This shows that SGC does have a convergence rate of $O(1/T)$, matching regular SGD \cite[Lemma~1]{SGDopt}. This shows that at least the convergence rate is not hurt by the fact that the data assignment is fixed.  However, unlike Theorem~\ref{thm:convl2}, this result does not always significantly improve as $d$ increases (although we note that the bound above \em is \em decreasing in $d_{\min}$, so in some parameter regimes---when $n \gg m$ and $p$ is close to $1$---this does indicate some improvement).  We leave it as an interesting open problem to fully generalize our result of Theorem~\ref{thm:convl2} to general loss functions.

%% file: sims.tex
\begin{table}
\scriptsize
\renewcommand{\arraystretch}{2}
\centering
\begin{tabular}{c|C{12cm}}
\footnotesize \bf Algorithm & \footnotesize \bf Brief description \\
\hline
Stochastic Gradient Code (SGC) & The master sends each data vector $\mathbf{x}_i$ to $d_i$ workers chosen at random, where $d_i$  is proportional to the $\ell_2$ norm of $\mathbf{x}_i$ and is computed as in~\eqref{eq:di} with $d=2$. Each worker sends to the master the weighted sum of its partial gradients as in~\eqref{eq:sumw}. The master computes the gradient estimate as the sum of the received results from non straggling workers.\\
Bernoulli Gradient Code (BGC) \cite{charles2017approximate} & Similar to SGC but all data vectors are replicated $d$ times, i.e., $d_i = 2$ for all $i\in [m]$.\\
{{\erasurehead}} \cite{charles2017approximate,wang2019erasurehead} & Partitions the data set equally and sends each partition to $d$ workers. Workers send the sum of the partial gradients to the master who computes the gradient estimate as the sum of distinct received partial gradients divided by total number of data vectors. \\
Ignore--Stragglers--SGD \cite{chen2016revisiting} & Partitions the data among the workers with no redundancy. Workers send the sum of the partial gradients to the master who computes the gradient estimate as the sum of distinct partial gradients divided by the average number of data vectors received per iteration. \\
SGC--Send--All & Same as SGC with one difference: at each iteration the workers send all the partial gradients to the master. The master computes the gradient estimate as the sum of \emph{distinct} partial gradients divided by the average number of data vectors received per iteration. \\ \hline
\end{tabular}
\caption{Summary of the stochastic algorithms that we implement in our simulations.}
\label{tab:summary}
\end{table}

\subsection{Simulation setup}
We simulated the performance of SGC on synthetic data $X$ of dimension $1000\times 100$. The data is generated as follows: each row vector $\mathbf{x}_i$ is generated using a Gaussian distribution $\mathcal{N}(0,100)$. We pick a random vector $\bar{\bbeta}$ with components being integers between $1$ and $10$ and generate $y_i \sim \mathcal{N}(\ip{\mathbf{x}_i}{\bar{\bbeta}},1)$. Our code and the generated data set can be found in \cite{SGCGit}.

We run linear regression using the $\ell_2$ loss function, i.e., $$\mathcal{L}(\mathbf{a}_i,\bbeta_t) = \dfrac{1}{2}\left(\ip{\mathbf{x}_i}{\bbeta_t} - y_i\right)^2.$$

\begin{figure}[b]
\setlength\figureheight{0.5\textwidth}
\setlength\figurewidth{0.7\textwidth}
\centering
\begin{minipage}{0.33\textwidth}
\resizebox{\textwidth}{!}{
\input{avg-over-10pb=0.1-w_-n=10,d=2,noise=1,datapoints=500,dim=50,iter=500str=0}
}
\captionsetup{subtype}
\caption{Error vs iterations, $p=0.1$.}
\label{fig:conv01}
\end{minipage}%
\hfill%
\begin{minipage}{0.33\textwidth}
\resizebox{\textwidth}{!}{
\input{avg-over-10pb=0.5-w_-n=10,d=2,noise=1,datapoints=500,dim=50,iter=500str=0.tex}
}
\captionsetup{subtype}
\caption{Error vs iterations, $p=0.7$.}
\label{fig:conv07}
\end{minipage}%
\hfill%
\begin{minipage}{0.33\textwidth}
\resizebox{\textwidth}{!}{
\input{convergence-n=10,d=2,noise=1,datapoints=500,dim=50,iter=500pow=0.7str=0.tex}
}
\captionsetup{subtype}
\caption{Error at $T=5000$ vs $p$.}
\label{fig:convp}
\end{minipage}
\caption{Convergence as function of probability of workers being stragglers $p$ is shown for small $p=0.1$ in~(a) and for large $p=0.7$ in~(b) for $n=10$ workers. SGC convergence has two phases: an exponential decay in the beginning until it reaches an error floor. SGC has same performance as BGC, but outperforms {\erasurehead} for large values of $p$. In~(c) the error floor at $T=5000$ iterations is shown versus $p$.}
\end{figure}

We show simulations for $n=10$ workers. For each simulation we vary the probability of a worker being a straggler from $p=0$ to $p=0.9$ with a step of $0.1$. We run the algorithm for $5000$ iterations with a variable step size given\footnote{In our theoretical analysis we assumed that the step size $\gamma_t$ is proportional to $1/t$. In our numerical simulations, we tried different  functions of  $\gamma_t$ and  observed that the one in \eqref{eq:gamma} gives better convergence rate for all the considered algorithms.} by %
\begin{equation}\label{eq:gamma}
\gamma_t = 7 \dfrac{\ln (10^{100})}{t^{0.7}}.\end{equation}
For all  simulations, we run each experiment $10$  times and average the results.  For SGC,  each data vector $\mathbf{x}_i$ is replicated $d_i$ times, where the $d_i$'s are computed as in~\eqref{eq:di} and~\eqref{eq:sigma} with $d=2$. Then, each $d_i$ is rounded to the nearest integer. Due to rounding, the actual value of $d$ given in~\eqref{eq:d} will be close to $2$. In our generated data set, the majority of the $d_i$'s are equal to $2$ while the others are either $1$ or $3$ resulting in average redundancy $d=2.024$. For the other algorithms in Table~\ref{tab:summary}, the average redundancy $d$ is chosen to be exactly equal to $2$.

We omit comparing SGC to the gradient codes   in \cite{raviv2017gradient} and \cite{horii2019distributed} because they do not match our setting; the former requires a high redundancy factor $d$ and the latter requires the master to run a decoding algorithm at each iteration.

\subsection{Convergence}
In Figures~\ref{fig:conv01} and~\ref{fig:conv07}, we plot the error $\norm{\bbeta_t-\bbeta^*}$ for up to $5000$ iterations for small and large probability of workers being stragglers, namely for $p=0.1$ and $p=0.7$, respectively. Here, $\bbeta^*$ given in ~\eqref{eq:optprob} is computed using the pseudoinverse of $X$, i.e., $\bbeta^* = \left(X^T X\right)^{-1}X^T Y$. We notice that for all $p$ the convergence rate of SGC exhibits two phases: an exponential decay followed by an error floor. To see the benefit of replication, we compare SGC to Ignore--Stragglers--SGD. Both have the same performance in the exponential phase, but SGC has a lower error floor due to redundancy. A lower bound on the performance of SGC is SGC--Send--All which has a lower error floor because it computes a better estimate of the gradient at the expense of a higher communication cost. However, as $p$ increases the gap between the two error floor of SGC and SGC--Send--All decreases. In our simulations, we notice the error floor of both algorithms almost match for $p\geq 0.6$ as can be seen in Figure~\ref{fig:convp}. 

For our chosen data set, SGC and BGC have similar performance. This is mainly due to the fact that most of the data vectors have the same replication factor in BGC and SGC which is $2$ times. For other data sets with more variance in the $d_i$'s, we observe that SGC can have better performance. {\erasurehead} has better error floor than SGC for small values of $p$. However, for large $p$ the rate of the exponential decay drastically decreases for {\erasurehead}.



\begin{figure}[b]
\setlength\figureheight{0.5\textwidth}
\setlength\figurewidth{0.7\textwidth}
\centering
\begin{minipage}{0.36\textwidth}
\resizebox{\textwidth}{!}{
\input{SGC-str-n=10.tex}
}
\captionsetup{subtype}
\caption{Error versus iterations, $p=0.7$.}
\label{fig:convstr}
\end{minipage}%
\hspace{1cm}
\begin{minipage}{0.36\textwidth}
\resizebox{\textwidth}{!}{
\input{errorfloor-n=10,d=2,noise=1,datapoints=500,dim=50,iter=500pow=0.7str=50.tex}
}
\captionsetup{subtype}
\caption{Error at iteration $T=5000$ versus $\nbiter$.}
\label{fig:convpstr}
\end{minipage}
\caption{The effect of the dependency of stragglers across iterations on the performance of SGC. We assume that the identity of the stragglers change every $\nbiter$ iterations. In~(a) the convergence of the error as function of the number of iterations is shown for different values of $\nbiter$ and $p=0.7$. SGC maintains an exponential decay in the error for the tested values of $\nbiter$ and $p$. However, the rate of the decay decreases and the error floor increases with the increase of $\nbiter$. In~(b), the error at iteration $T=5000$ as function of $\nbiter$ is shown for different values of $p$.}
\label{fig:con}
\end{figure}

\subsection{Dependency between stragglers across iterations}
Our theoretical analysis assumes that the stragglers are independent across iterations. We check the effect of this dependency on the numerical performance of SGC.  We use a simple model to enforce  dependency of stragglers across iterations by fixing the stragglers for   $\nbiter$ iterations, after which the stragglers are chosen again randomly and iid with a probability $p$ and this is repeated until the algorithm stops.
The special value of   $\nbiter =1$ implies that the stragglers are independent. A large value of $\nbiter$ implies a longer dependency among the stragglers across iterations. 
We observe in Figure~\ref{fig:con} that SGC still maintains the two phases behavior. However, as $\nbiter$ increases, the rate of convergence decreases and the error floor increases. 

%% file: convergence-n=10,d=2,noise=1,datapoints=500,dim=50,iter=500pow=0.7str=0.tex
\begin{tikzpicture}

\definecolor{color0}{rgb}{0.12156862745098,0.466666666666667,0.705882352941177}
\definecolor{color1}{rgb}{1,0.498039215686275,0.0549019607843137}
\definecolor{color2}{rgb}{0.172549019607843,0.627450980392157,0.172549019607843}
\definecolor{color3}{rgb}{0.83921568627451,0.152941176470588,0.156862745098039}
\definecolor{color4}{rgb}{0.580392156862745,0.403921568627451,0.741176470588235}
\definecolor{color5}{rgb}{0.549019607843137,0.337254901960784,0.294117647058824}
\definecolor{color6}{rgb}{0.890196078431372,0.466666666666667,0.76078431372549}
\begin{axis}[
title style = {font=\Large},
xlabel={Probability of straggling $p$},
ylabel={$\norm{\bbeta_T - \bbeta^*}$},
ylabel style = {font = \Large, at={(-0.02,0.5)}},
xlabel style = {font=\Large},
xmin=-0.045, xmax=0.945,
ymin=1.86957061738095e-08, ymax=0.0918564790229895,
ymode = log,
width=\figurewidth,
height=\figureheight,
tick align=outside,
tick pos=left,
x grid style={white!69.01960784313725!black},
y grid style={white!69.01960784313725!black},
legend style={at={(0.95,0.4)}, anchor=north east, draw=white!80.0!black},
legend cell align={left}
]

\def\msize{2}

\addplot [ultra thick, color1, dashed, mark=*, mark size=\msize, mark options={solid}]
table {%
0 3.76612708530429e-08
0.1 3.64333587466363e-05
0.2 5.25404994195889e-05
0.3 7.7615181021457e-05
0.4 8.82456978241633e-05
0.5 9.3760156052778e-05
0.6 0.000124319114537291
0.7 0.000145097726659664
0.8 0.000211614411280793
0.9 0.000286440597633678
};
\addlegendentry{Ignore--Stragglers--SGD}

\addplot [ultra thick, color5, mark=asterisk, mark size=\msize, mark options={solid}]
table {%
0 3.76612708530429e-08
0.1 2.3200739209075e-05
0.2 3.4045487330212e-05
0.3 4.84481208579209e-05
0.4 5.94629155464262e-05
0.5 5.83818396013618e-05
0.6 8.16981993717e-05
0.7 9.92088779126475e-05
0.8 0.00014090527548098
0.9 0.000192971252083702
};
\addlegendentry{BGC}

\addplot [ultra thick, color0, mark=square*, mark size=\msize, mark options={solid}]
table {%
0 3.76612708530429e-08
0.1 2.50850184492928e-05
0.2 3.56732644460501e-05
0.3 5.24751138616398e-05
0.4 6.09880379720325e-05
0.5 5.89368728831169e-05
0.6 8.56997507201122e-05
0.7 9.82836711012897e-05
0.8 0.000140529513791184
0.9 0.000193389169903699
};
\addlegendentry{SGC}

\addplot [ultra thick, color3, mark = +, mark size = \msize, mark options = {solid}]
table {%
0 3.76612708530429e-08
0.1 9.64861116289711e-06
0.2 1.93644890211481e-05
0.3 2.69691399160198e-05
0.4 4.14381715419961e-05
0.5 4.57826355068105e-05
0.6 7.431347568075e-05
0.7 0.000394148935463102
0.8 0.00466606686584583
0.9 0.0455991457852139
};
\addlegendentry{\sc ErasureHead}

\addplot [ultra thick, color2, dotted, mark=triangle*, mark size=\msize, mark options={solid}]
table {%
0 3.76612708530429e-08
0.1 1.13455437312175e-05
0.2 2.22593341347114e-05
0.3 3.64810012708523e-05
0.4 4.75076162224916e-05
0.5 5.31461283261025e-05
0.6 7.93040356830945e-05
0.7 9.14665893115764e-05
0.8 0.000135398638947876
0.9 0.000185873582618849
};
\addlegendentry{SGC--Send--All}

\end{axis}

\end{tikzpicture}

%% file: errorfloor-n=10,d=2,noise=1,datapoints=500,dim=50,iter=500pow=0.7str=50.tex
\begin{tikzpicture}

\definecolor{color0}{rgb}{0.12156862745098,0.466666666666667,0.705882352941177}
\definecolor{color1}{rgb}{1,0.498039215686275,0.0549019607843137}
\definecolor{color2}{rgb}{0.172549019607843,0.627450980392157,0.172549019607843}
\definecolor{color3}{rgb}{0.83921568627451,0.152941176470588,0.156862745098039}
\definecolor{color4}{rgb}{0.580392156862745,0.403921568627451,0.741176470588235}
\definecolor{color5}{rgb}{0.549019607843137,0.337254901960784,0.294117647058824}
\definecolor{color6}{rgb}{0.890196078431372,0.466666666666667,0.76078431372549}
\definecolor{color7}{rgb}{0.23, 0.27, 0.29}
\definecolor{color8}{rgb}{0.0, 0.5, 1.0}
\definecolor{color9}{rgb}{0.74, 0.2, 0.64}
\definecolor{color10}{rgb}{0.48, 0.25, 0.0}

\begin{axis}[
title style = {font=\Large},
xlabel={Iterations until stragglers change, $\nbiter$},
ylabel={$\norm{\bbeta_T - \bbeta^*}$},
ylabel style = {font = \Large, at={(-0.0,0.5)}},
xlabel style = {font=\Large},
xmin=0, xmax=1000,
ymin=2.57688445014917e-05, ymax=0.00587657409386,
width=\figurewidth,
height=\figureheight,
tick align=outside,
tick pos=left,
x grid style={white!69.01960784313725!black},
y grid style={white!69.01960784313725!black},
legend style={at={(0.01,0.99)}, anchor=north west, draw=white!80.0!black},
legend cell align={left},
cycle list name=color,
]

\def\msize{2}


\addplot [ultra thick, color4, dashed, mark=circle, mark size=3, mark options={solid}]
table {%
1	0.000140529513791184
50	0.00114342073021
100	0.00122370286019539
500	0.00411040263875965
1000 0.00587657409386
};
\addlegendentry{$p = 0.8$}

\addplot [ultra thick, color3, mark=*, mark size=3, mark options={solid}]
table {%
1	9.82836711012897e-05
50	0.000754023266529
100	0.00101493905339174
500	0.00225276908391849
1000 0.00309788341776
};
\addlegendentry{$p = 0.7$}


\addplot [ultra thick, color2, dotted, mark=square, mark size=3, mark options={solid}]
table {%
1	5.89368728831169e-05
50	0.000470276168537
100	0.000632886713728858
500	0.00147080641387916
1000 0.00150407182499
};
\addlegendentry{$p = 0.5$}


\addplot [ultra thick, color1, dashed, mark=asterisk, mark size=3, mark options={solid}]
table {%
1	5.24751138616398e-05
50	0.000324996814166
100	0.000475856165720481
500	0.000948715897046086
1000 0.00116148465175
};
\addlegendentry{$p = 0.3$}


\addplot [ultra thick, color0, mark=+, mark size=3, mark options={solid}]
table {%
1	2.50850184492928e-05
50	0.000164140691289
100	0.000191371424454783
500	0.000441028808398464
1000	0.000317455166232
};
\addlegendentry{$p = 0.1$}


\end{axis}

\end{tikzpicture}

%% file: wtd_proof.tex
In this section, we prove Theorem~\ref{thm:convl2}.
In the case when the loss function is 
\[ \mathcal{L}([X|y], \bbeta) = \dfrac{1}{2} \| X\bbeta - \mathbf{y} \|_2^2, \]
we have
\[ \nabla \mathcal{L}(\ba_i , \bbeta) = (\ip{\bx_i}{\bbeta} - y_i) \cdot \bx_i, \]
so that
\[ \sum_i ^m\nabla \mathcal{L}(\ba_i, \bbeta)  = X^T(X\bbeta - \mathbf{y}) = \nabla \mathcal{L}(A, \beta).\]

Fix an iteration $t$.
Let $Z_i$ (which depends on $t$; we suppress this dependence in the notation) be defined by
\[ Z_i = \sum_{j=1}^n \mathcal{I}_i^j.\]
That is, $Z_i$ is the number of workers who hold $\ba_i$ who are not stragglers at round $t$. 
Thus, $Z_i$ is a binomial random variable with mean $d_i(1-p)$ and variance $d_i p (1 - p).$
Let 
\[ \tZ_i = Z_i - \EE Z_i, \]
so that
\[ \EE (\tZ_i)^2 = d_i p (1 - p) \]
and 
\[ \EE \tZ_i \tZ_j = \frac{ d_i d_j }{n} p (1 - p). \]

From the definition of $\bbeta_{t+1}$ and replacing $\hat{\mathbf{g}}_t$ by its value from~\eqref{eq:ghat}, we have
\begin{align*}
 \bbeta_{t+1} &= \bbeta_t - \gamma_t \hat{\mathbf{g}}_t\\
&= \bbeta_t - \gamma_t \sum_{j=1}^n \sum_{i=1}^m \frac{ \mathcal{I}_i^j }{ d_i (1 - p) } \nabla \mathcal{L}(\mathbf{a}_i,\bbeta) \\  
&= \bbeta_t - \gamma_t \sum_{i=1}^m \frac{ Z_i }{ d_i (1 - p) } ( \ip{\bx_i}{\bbeta_t} - y_i) \bx_i \\
&= \bbeta_t - \sum_{i=1}^m Z_i \delta_i ( \ip{\bx_i}{\bbeta_t} - y_i) \bx_i,
\end{align*}
where we define
\[ \delta_i = \frac{ \gamma_t }{ d_i (1 - p ) }. \]
(Notice that $\delta_i$ also depends on $t$; we suppress this for notational convenience).

Expanding out the terms, we have
\begin{align*}
\bbeta_{t+1} - \bbeta^* 
&= \bbeta_t - \bbeta^* - \sum_{i=1}^m \EE Z_i \delta_i \bx_i \bx_i^T (\bbeta_t - \bbeta^*)\\
&\qquad\qquad- \sum_{i=1}^m \tZ_i \delta_i \bx_i \bx_i^T(\bbeta_t - \bbeta^*) \\
&\qquad\qquad- \sum_{i=1}^m\EE Z_i \delta_i (\ip{\bx_i}{\bbeta^*} - y_i) \bx_i \\
&\qquad\qquad- \sum_{i=1}^m \tZ_i \delta_i (\ip{\bx_i}{\bbeta^*} - y_i) \bx_i
\end{align*}
where we have split up $Z_i = \EE Z_i + \tZ_i$ and 
\[ (\ip{\bx_i}{\bbeta_t} - y_i) \bx_i = \bx_i \bx_i^T (\bbeta_t - \bbeta^*) + (\ip{\bx_i}{\bbeta^*} - y_i)\bx_i.\]
Letting
\[ \mathbf{r} := X \bbeta^* - \mathbf{y} \]
be the optimal residual and writing the above in matrix notation, we have
\[ \bbeta_{t+1} - \bbeta^* = (\bbeta_t - \bbeta^*) - (1 - p )X^T D_d D_\delta X (\bbeta_t - \bbeta^*)
- X^T D_{\tZ} D_\delta X (\bbeta_t - \bbeta^*) 
- (1 - p) X^T D_d D_\delta \mathbf{r}
- X^T D_{\tZ} D_\delta \mathbf{r}, \]
where $D_{\tZ}$ is diagonal with entries $\tZ_i$, $D_\delta$ is diagonal with entries $\delta_i$, and $D_d$ is diagonal with entries $d_i$.
Recalling that 
\[ \delta_i = \frac{\gamma_t}{(1-p)d_i} \] 
we have
\[ D_\delta \cdot D_d = \frac{\gamma_t}{1-p} . \]
Thus we can simplify the above as 
\begin{align*}
\bbeta_{t+1} - \bbeta^* &= 
(\bbeta_t - \bbeta^*) - \gamma_t X^T  X (\bbeta_t - \bbeta^*)
- X^T D_{\tZ} D_\delta X (\bbeta_t - \bbeta^*) 
- \gamma_t X^T  \mathbf{r}
- X^T D_{\tZ} D_\delta \mathbf{r} \\
&= (\bbeta_t - \bbeta^*) - \gamma_t X^T  X (\beta_t - \beta^*)
- X^T D_{\tZ} D_\delta X (\bbeta_t - \bbeta^*) 
- X^T D_{\tZ} D_\delta \mathbf{r}
\end{align*}
using the fact that $X^T \mathbf{r} = 0$ since $\mathbf{r} = X\bbeta^* - \mathbf{y}$ is the optimal residual.
We simplify this further as:
\[ \bbeta_{t+1} - \bbeta^* =
(I - \gamma_t X^T X + X^T D_{\tZ} D_\delta X)(\bbeta_t - \bbeta^*) 
- X^T D_{\tZ} D_\delta \mathbf{r}. \]

Now we compute $\EE\|\bbeta_{t+1} - \bbeta^*\|^2$, where the expectation is over the choice of $\beta_{t+1}$, conditioned on $\beta_t$.  We have
\begin{align}
\EE \| \bbeta_{t+1} - \bbeta^* \|^2
&= (\bbeta_t - \bbeta^*)^T \left[ (I - \gamma_t X^T X)^2 +  X^T D_\delta \EE[D_{\tZ} X X^T D_{\tZ}] D_\delta X \right](\bbeta_t - \bbeta^*) \label{eq:1} \\
& \qquad + \mathbf{r}^T D_\delta \EE [D_{\tZ} X X^T D_{\tZ}]  D_\delta \mathbf{r}  \label{eq:2} \\
& \qquad + \mathbf{r}^T D_\delta \EE [D_{\tZ} X X^T D_{\tZ} ]D_\delta X (\bbeta_t - \bbeta^*) \label{eq:3} \\
& \qquad + \mathbf{r}^T D_\delta \EE D_{\tZ} XX^T( I - \gamma_t X^T X ) ( \bbeta_t - \bbeta^* ) \label{eq:4}\\
& \qquad + (\bbeta_t - \bbeta^*)^T (I - \gamma_t X^T X)( X^T \EE D_{\tZ} D_\delta X ) ( \bbeta_t - \bbeta^* ). \label{eq:5}
\end{align}
We handle each of these terms below.
First, we observe that \eqref{eq:4} and \eqref{eq:5} are zero because $\EE D_{\tZ} = 0$.
In order to handle \eqref{eq:1}, \eqref{eq:2}, \eqref{eq:3},
we compute
 \[ \EE D_{\tZ} X X^T D_{\tZ}. \]
The off-diagonal elements are given by
\[ \EE \tZ_i \tZ_j \ip{\bx_i}{\bx_j} = \frac{d_id_j}{n}p(1-p) \ip{\bx_i}{\bx_j}, \]
and the diagonal elements are given by
\[ \EE \tZ_i^2 \|\bx_i\|^2 = d_i p (1-p) \|\bx_i\|^2 = d_i^2 p(1-p) / \sigma .\]
Thus, 
\[ \EE D_{\tZ} X X^T D_{\tZ} = p(1-p) \inparen{ \frac{1}{n} D_d XX^T D_d + \frac{1}{\sigma}\inparen{I -\frac{D_d}{n}} D_d^2 }. \]
Now we handle the terms \eqref{eq:1} and \eqref{eq:2}.
First, for \eqref{eq:1}, we have
\begin{align*}
&(\bbeta_t - \bbeta^*)^T \left[ (I - \gamma_t X^T X)^2 + \EE X^T D_\delta D_{\tZ} X X^T D_{\tZ} D_\delta X \right](\bbeta_t - \bbeta^*) \\
&= (\bbeta_t - \bbeta^*)^T \left[ (I - \gamma_t X^T X)^2 + p(1-p)X^T D_\delta \inparen{ \frac{1}{n} D_d XX^T D_d + \frac{1}{\sigma} \inparen{ I - \frac{D_d}{n}} D_d^2 } D_\delta X \right] (\bbeta_t - \bbeta^*) \\
&= (\bbeta_t - \bbeta^*)^T \left[ (I - \gamma_t X^T X)^2 + \frac{p(1-p)}{n}X^T D_\delta D_d XX^T D_dD_\delta X
+ \frac{p(1-p)}{\sigma}X^T \inparen{ I - \frac{D_d}{n} }D_\delta D_d^2 D_\delta X \right] (\bbeta_t - \bbeta^*) \\
&= (\bbeta_t - \bbeta^*)^T \left[ (I - \gamma_t X^T X)^2 + \frac{p(1-p)\gamma_t^2}{n}X^T  XX^T X
+ \frac{\gamma_t^2 p}{(1-p)\sigma}X^T \inparen{ I - \frac{D_d}{n}}  X \right] (\bbeta_t - \bbeta^*),
\end{align*}
where in the last line we used the fact again that $D_\delta D_d = \gamma_t I/(1-p)$.  Now we can bound this term by
\[ \eqref{eq:1} \leq 
\inparen{(1 - \gamma_t \|X^T X\|)^2 + \frac{ p \gamma_t^2 }{(1-p)n} \|X^T X\|^2 + \frac{\gamma_t^2 p}{(1-p)\sigma} \|X^T X\|}\|\bbeta_t - \bbeta^*\|^2, \]
where above we have used the fact that
\[ \left\| X^T \inparen{ I -\frac{D_d}{n} } X \right\| \leq \| X^T X \|, \]
because $I - D_d/n$ is a diagonal matrix whose diagonal entries are all in $[0,1]$ (using the fact that $d_i \leq n$ for all $i$).
The second term \eqref{eq:2} is bounded by
\begin{align*}
\eqref{eq:2} &=  \mathbf{r}^T D_\delta \EE D_{\tZ} X X^T D_{\tZ}  D_\delta \mathbf{r} \\
&= p(1-p)\mathbf{r}^T D_\delta  \inparen{ \frac{1}{n} D_d XX^T D_d + \inparen{ I - \frac{D_d}{n} }\frac{1}{\sigma} D_d^2 } D_\delta \mathbf{r} \\
&= \frac{p(1-p)}{n}\mathbf{r}^T D_\delta  D_d XX^T D_d D_\delta \mathbf{r} + \frac{p(1 - p)}{\sigma} \mathbf{r}^T \inparen{ I - \frac{D_d}{n} } D_\delta D_d^2 D_\delta \mathbf{r} \\
&\leq \frac{ \gamma_t^2 p }{(1-p)n} \mathbf{r}^T XX^T \mathbf{r} + \gamma_t^2 \cdot \frac{p}{1-p}\cdot \frac{ \mathbf{r}^T \inparen{ I - \frac{D_d}{n} } \mathbf{r} }{\sigma} \\
&\leq \gamma_t^2 \cdot \frac{p}{1-p} \cdot \frac{\|\mathbf{r}\|^2}{\sigma}, 
\end{align*}
where we have used the fact that $X^T \mathbf{r} = 0$, and that $\mathbf{r}^T\inparen{I - D_d/n} \mathbf{r} \leq \|\mathbf{r}\|^2$ because $d_i \leq n$ for all $i$.

Finally we bound \eqref{eq:3}.  We have, using our expression for $\EE[ D_{\tZ} X X^T D_{\tZ} ]$ from above that
\begin{align*} \eqref{eq:3} &=  \mathbf{r}^T D_\delta \EE [D_{\tZ} X X^T D_{\tZ} ]D_\delta X (\bbeta_t - \bbeta^*) \\
&= p(1-p) \mathbf{r}^T D_\delta  \inparen{ \frac{1}{n} D_d XX^T D_d + \frac{1}{\sigma}\inparen{I -\frac{D_d}{n}} D_d^2 } D_\delta X (\bbeta_t - \bbeta^* ) \\
 &= \frac{p(1 - p)}{n} \mathbf{r}^T D_\delta D_d XX^T D_d D_\delta X (\bbeta_t - \bbeta^*) + \frac{p(1-p)}{\sigma} \mathbf{r}^T D_\delta \inparen{ I - \frac{D_d}{n} } D_d^2 D_\delta X ( \bbeta_t - \bbeta^* ) \\
 &= \frac{\gamma_t^2 p}{(1-p)n} \mathbf{r}^T  XX^T  X (\bbeta_t - \bbeta^*) + \frac{\gamma_t^2 p}{(1-p)\sigma} \mathbf{r}^T \inparen{ I - \frac{D_d}{n} } X ( \bbeta_t - \bbeta^* ) 
\end{align*}
using the fact that $D_\delta D_d = \gamma_t I/(1-p)$ in the last line.  Now, the first term is equal to zero because $\mathbf{r}^T X = 0$, and we have
\begin{align*}
\eqref{eq:3} 
&= \frac{\gamma_t^2 p}{(1-p)\sigma} \mathbf{r}^T \inparen{ I - \frac{D_d}{n} } X ( \bbeta_t - \bbeta^* ) \\
&= \frac{ \gamma_t^2 p} { (1-p)\sigma n  } \mathbf{r}^T D_d X (\bbeta_t - \bbeta^* )
\end{align*}
again using the fact that $\mathbf{r}^T X = 0$.  Finally, we can bound
\begin{align*}
\eqref{eq:3} 
&= \frac{ \gamma_t^2 p } { (1-p)\sigma n } \mathbf{r}^T D_d X (\bbeta_t - \bbeta^* )  \\
&\leq \frac{ \gamma_t^2 p  }{(1-p)\sigma n} \| \mathbf{r} \| \| D_d X ( \bbeta_t - \bbeta^* ) \| \\
&\leq \frac{ \gamma_t^2 p }{(1-p)\sigma } \| \mathbf{r} \| \| X ( \bbeta_t - \bbeta^* ) \| \\
&\leq \frac{ \gamma_t^2 p }{(1-p)\sigma } \| \mathbf{r} \| \sqrt{ \|X^T X\| } \| \bbeta_t - \bbeta^* \| \\
&\leq \frac{ \gamma_t^2 p}{(1-p)\sigma} \inparen{ \frac{ \|\mathbf{r}\|^2 + \|X^T X \| \| \bbeta_t - \bbeta^* \|^2 }{2} },
\end{align*}
using the arithmetic-geometric-mean inequality in the final line.  In particular, this term is similar to terms that appear in both \eqref{eq:1} and \eqref{eq:2}, and (along with the observation that \eqref{eq:4}, \eqref{eq:5} are zero) we have
\begin{align}
\EE \| \bbeta_{t+1} - \bbeta^* \|^2
&\leq \eqref{eq:1} + \eqref{eq:2} + \eqref{eq:3} \notag \\
&\leq \inparen{(1 - \gamma_t \|X^T X\|)^2 + \frac{ p \gamma_t^2 }{(1-p)n} \|X^T X\|^2 + \frac{\gamma_t^2 p}{(1-p)\sigma} \|X^T X\|}\|\bbeta_t - \bbeta^*\|^2 \notag \\
&\qquad + \gamma_t^2 \cdot \frac{p}{1-p} \cdot \frac{\|\mathbf{r}\|^2}{\sigma} \notag\\
&\qquad +  \frac{ \gamma_t^2 p }{(1-p)\sigma} \inparen{ \frac{ \|\mathbf{r}\|^2 + \|X^T X \| \| \bbeta_t - \bbeta^* \|^2 }{2} }\notag \\
&\leq 
\inparen{(1 - \gamma_t \|X^T X\|)^2 + \frac{ p \gamma_t^2 }{(1-p)n} \|X^T X\|^2 + \frac{2\gamma_t^2 p}{(1-p)\sigma} \|X^T X\|}\|\bbeta_t - \bbeta^*\|^2 \label{eq:1prime}\\
&\qquad + 2\gamma_t^2 \cdot \frac{p}{1-p} \cdot  \frac{\|\mathbf{r}\|^2}{\sigma}. \label{eq:2prime} 
\end{align}

Now we recall our choice of 
\[ \gamma_t = \frac{1}{\|X^T X \|} \cdot \min\inset{ \frac{1}{2} , \frac{ \log(1/\eps^2) }{ t } }, \]
and the definition of 
\[ \sigma = \frac{d}{\mu} \frac{1}{\|X^T X\|}. \]
Let
\[ \tgt := \|X^T X\|\gamma_t = \min \inset{ \frac{1}{2} , \frac{ \log(1/\eps^2 }{t} }. \]

Now we can simplify our bounds on \eqref{eq:1prime} and \eqref{eq:2prime} as:
\begin{align*}
\eqref{eq:1prime} &\leq 
\inparen{(1 - \gamma_t \|X^T X\|)^2 + \frac{ p \gamma_t^2 }{(1-p)n} \|X^T X\|^2 + 2\frac{\gamma_t^2 p}{(1-p)\sigma}  \|X^T X\|}\|\bbeta_t - \bbeta^*\|^2 \\  
&\leq \inparen{ \inparen{ 1 - \tgt }^2 + \inparen{ \frac{p}{1-p} } \inparen{ \tgt }^2 \cdot \frac{1}{n} + \inparen{ \frac{2p}{1-p} } \inparen{ \tgt }^2 \inparen{ \frac{\mu}{d} }}\| \bbeta_t - \bbeta^*\|^2\\
&\leq \inparen{ \inparen{ 1 - \tgt }^2 + \frac{1}{2} \inparen{ \tgt }^2 } \|\beta_t - \beta^*\|^2,
\end{align*}
using the assumptions that $n \geq 4p/(1 - p)$ and $d \geq 8\mu p/(1-p)$.
Now we have:
\begin{align*}
\eqref{eq:1prime}
&\leq \inparen{ \inparen{ 1 - \tgt }^2 + \frac{1}{2} \inparen{ \tgt }^2 } \|\bbeta_t - \bbeta^*\|^2 \\ 
&= \inparen{  1 - 2\tgt  + \frac{3}{2} \tgt } \|\bbeta_t - \bbeta^*\|^2 \\
&\leq \inparen{  1 - \tgt  } \|\bbeta_t - \bbeta^*\|^2,
\end{align*}
using from the definition of $\tgt$ that $\tgt \leq 1/2$ and hence $\tgt^2 \leq \frac{1}{2} \tgt$.

Meanwhile, 
\begin{align*}
\eqref{eq:2prime} 
&\leq 2\gamma_t^2 \cdot \frac{p}{1-p} \frac{\|\mathbf{r}\|^2}{\sigma} \\
&\leq 2\inparen{ \tgt}^2 \inparen{ \frac{p}{1 - p} } \frac{ \|\mathbf{r}\|^2 }{ \sigma \|X^T X \|^2 }\\
&\leq 2\inparen{ \tgt}^2 \inparen{ \frac{p}{1 - p} } \frac{ \|\tilde{\mathbf{r}}\|^2 }{ \sigma \|X^T X \| },
\end{align*}
recalling that $\tilde{\mathbf{r}} = \mathbf{r}/\|X^T X\|$.  Thus
\begin{align*}
\eqref{eq:2prime} 
&\leq 2\inparen{ \tgt}^2 \inparen{ \frac{p}{1 - p} } \frac{ \|\tilde{\mathbf{r}}\|^2 }{ \sigma \|X^T X \| } \\
&= 2\inparen{ \tgt }^2 \inparen{ \frac{p}{1 - p} } \inparen{ \frac{\mu}{d} } \|\tilde{\mathbf{r}}\|^2.
\end{align*}
Putting the two terms together, we conclude that for fixed $t$, 
\[ \EE\|\bbeta_{t+1} - \bbeta^*\|_2^2 \leq \inparen{ 1 - \tgt } \|\bbeta_t - \bbeta^*\|_2^2 + 
 2\inparen{ \tgt }^2 \inparen{ \frac{p}{1 - p} } \inparen{ \frac{\mu}{d} } \|\tilde{\mathbf{r}}\|^2. \]

Now, we proceed by induction, 
using the fact that the stragglers are independent between the different rounds, to conclude that
\begin{align*}
\EE\|\beta_{T} - \beta^*\|_2^2 
&\leq \inparen{ \prod_{t=1}^T ( 1 - \tgt ) }  \|\bbeta_0 - \bbeta^* \|_2^2 + 2T \tilde{\gamma}^2_T \inparen{ \frac{p}{1-p} } \inparen{ \frac{\mu}{d} } \| \tilde{\mathbf{r}} \|_2^2 \\
&\leq \inparen{1 - \frac{\log(1/\eps^2)}{T} }^T \| \bbeta_0 - \bbeta^* \|_2^2 + 2T \cdot \frac{ \log^2(1/\eps^2) }{T^2 } \inparen{ \frac{p}{1-p} } \inparen{ \frac{\mu}{d} } \| \tilde{\mathbf{r}} \|_2^2 \\
&\leq \eps^2 \| \bbeta_0 - \bbeta^* \|_2^2 +  \frac{2}{Td} \cdot \inparen{ \log^2(1/\eps^2) \inparen{ \frac{p}{1-p} } \mu \|\tilde{\mathbf{r}}\|_2^2 }.
\end{align*}
This proves the theorem.

%% file: analysis.tex
In this section we 
prove Theorem~\ref{thm:conv}.
Our proof of Theorem~\ref{thm:conv} relies on the following result from \cite{SGDopt}
which shows that any stochastic algorithm with a ``good'' estimator of the true gradient converges with rate $\mathcal{O}(\frac{1}{T})$.  We translate this result to our setting.
\begin{theorem}[Lemma~1 in \cite{SGDopt}]\label{thm:sgdopt}
Suppose $\mathcal{L}$ is $\lambda$-strongly convex with respect to $\bbeta^*$ over a convex set $\mathcal{W}$, and that $\ghatt$ is an unbiased estimator of a subgradient $\mathbf{g}_t$ of the loss function $\mathcal{L}$ at $\bbeta_t$, i.e., $\EE \ghatt = \mathbf{g}_t$.  Suppose also that for all $t$, $\EE\norm{\ghatt}_2^2 \leq G$.\footnote{  Here, the randomness in the expectation is over the next round of stragglers, conditioned on the previous rounds.}  Then if we pick $\gamma_t = 1/\lambda t$, it holds for any $T$ that
\bes
\EE \norm{\bbeta_T - \bbeta^*}_2^2 \leq \dfrac{4 G}{\lambda^2 T}.
\ees
\end{theorem}

\begin{proof}[Proof of Theorem~\ref{thm:conv}]
In order to apply Theorem~\ref{thm:sgdopt}, we need to show that the estimate of the gradient obtained by the master at each iteration is unbiased.
To see this, recall that at each iteration $t$, the master computes the following estimate of the gradient:
\be
\hat{\mathbf{g}}_t \triangleq \sum_{j = 1}^{n} \sum_{i=1}^{m} \dfrac{ \mathcal{I}_i^j }{d_i(1-p)}\grad{i}{t}.
\ee
Therefore, 
\begin{align}
\mathbb{E}\hat{\mathbf{g}}_t & = \sum_{j = 1}^{n} \sum_{i=1}^{m} \dfrac{ \EE\mathcal{I}_i^j }{d_i(1-p)}\grad{i}{t}.
\end{align}
Recall that $\mathcal{I}_i^j$ is an indicator function equal to $1$ if worker $j$ is non straggler and has data vector $\mathbf{a}_i$. 
Thus,
\begin{align*}
\EE \hat{ \mathbf{g}}_t &=  \sum_{j=1}^n \sum_{i=1}^m \frac{ \ind{ \text{ worker $j$ has data vector $\mathbf{a}_i$ } } }{ d_i } \grad{i}{t} \\
&=  \sum_{i=1}^m \grad{i}{t} \\
&= \nabla \mathcal{L}(A, \bbeta_t ).
\end{align*}

Now, we need to show that under the conditions of the theorem, the variance $\EE \|\hat{\bg}(A, \bbeta_t)\|_2^2$ is bounded.  (Here, the randomness is over the choice of the stragglers in round $t$).
As in the proof of Theorem~\ref{thm:convl2}, let $Z_i$ be the binomial random variable that counts the number of non-stragglers (in a given round $t$) who have block $i$.  Thus we have
$$\EE \left(Z_i^2\right)  = d_ip(1-p) + d_i^2(1-p)^2, \qquad \qquad \EE Z_{i_1}Z_{i_2} = \dfrac{d_{i_1} d_{i_2}}{n} p(1-p) + d_{i_1} d_{i_2}(1-p)^2.$$

We compute
\begin{align*}
\EE \twonorm{ \hat{ \bg}( A, \bbeta_t) }^2 &\leq
\EE \max_{ \bbeta \in \mathcal{W} } \twonorm{ \hat{\bg}(A, \bbeta ) }^2 \\
&= \EE \max_{ \bbeta \in \mathcal{W} } \twonorm{ \sum_{j=1}^n \sum_{i=1}^m \frac{ \mathcal{I}_i^j } { d_i (1-p) } \nabla \mathcal{L}( \ba_i, \bbeta ) }^2 \\
&= \frac{1}{(1-p)^2}\EE \max_{ \bbeta \in \mathcal{W} } \twonorm{ \sum_{i=1}^m \frac{ Z_i }{ d_i } \nabla\mathcal{L}( \ba_i, \bbeta ) }^2\\
&= \frac{1}{(1-p)^2} \EE \max_{ \bbeta \in \mathcal{W} }\sum_{i_1 = 1}^m \sum_{i_2 = 1}^m \frac{ Z_{i_1} Z_{i_2} }{d_{i_1}d_{i_2} } \ip{ \nabla \mathcal{L}( \ba_{i_1} , \bbeta )}{ \nabla \mathcal{L} (\ba_{i_2}, \bbeta) } \\
&\leq \frac{1}{(1-p)^2} \sum_{i_1 = 1}^m \sum_{i_2 = 1}^m \frac{ \EE[ Z_{i_1} Z_{i_2}]  }{d_{i_1}d_{i_2} } \max_{\bbeta \in \mathcal{W} }\ip{ \nabla \mathcal{L}( \ba_{i_1} , \bbeta )}{ \nabla \mathcal{L} (\ba_{i_2}, \bbeta) },
\end{align*}
where above we have used the fact that the terms $\EE[ Z_{i_1} Z_{i_2} ] / (d_{i_1} d_{i_2})$ are all positive to move the maximum inside the sum.
We have
\[ \ip{ \nabla \mathcal{L}( \ba_{i_1} , \bbeta )}{ \nabla \mathcal{L} (\ba_{i_2}, \bbeta) } \leq \|\nabla \mathcal{L}( \ba_{i_1} , \bbeta )\|_2 \| \nabla \mathcal{L}( \ba_{i_2} , \bbeta ) \|_2 \] 
by Cauchy-Shwarz, and thus
\[ \max_{\bbeta \in \mathcal{W}} \ip{ \nabla \mathcal{L}( \ba_{i_1} , \bbeta )}{ \nabla \mathcal{L} (\ba_{i_2}, \bbeta) } \leq \max_{i\in[n]} \max_{ \bbeta \in \mathcal{W} }\|\nabla \mathcal{L}( \ba_{i} , \bbeta )\|^2_2 \leq C^2, \]
by the assumptions of the theorem.
Thus, we may continue the derivation above as
\begin{align*}
\EE \twonorm{ \hat{ \bg}( A, \bbeta_t) }^2 &\leq
\frac{C^2}{(1-p)^2} \sum_{i_1 = 1}^m \sum_{i_2 = 1}^m \frac{ \EE[ Z_{i_1} Z_{i_2}]  }{d_{i_1}d_{i_2} }  \\
&= \frac{C^2}{(1-p)^2} \inparen{ \sum_{i=1}^m \inparen{ \frac{ d_i p(1-p) + d_i^2 (1-p)^2 }{d_i^2 } } + \sum_{i_1 \neq i_2} \inparen{ \frac{ d_{i_1}d_{i_2} p(1-p) }{n d_{i_1}d_{i_2} } + \frac{ d_{i_1} d_{i_2} (1-p)^2 }{d_{i_1}d_{i_2} } } } \\
&= C^2 \inparen{ \sum_{i=1}^m \inparen{ \frac{  p}{ (1-p) d_i}  +  1 } + \sum_{i_1 \neq i_2} \inparen{ \frac{  p }{n (1-p)} + 1 } } \\ 
&\leq  mC^2 \inparen{ \frac{p}{1-p} \cdot \frac{1}{d_{\min} } +  \frac{ (m-1) p }{n (1-p) } + m } 
\end{align*}
Plugging this estimate into Theorem~\ref{thm:sgdopt} proves Theorem~\ref{thm:conv}.

\end{proof}

%% file: related.tex
In this section we survey the related work more broadly than in the introduction.
\paragraph{Coding techniques for straggler mitigation}
Straggler workers are the bottleneck of distributed systems and mitigation of stragglers is a must \cite{DB13}. Amongst popular techniques, coding theoretic techniques are being used for straggler mitigation in different applications such as machine learning, see e.g. \cite{tandon2017gradient, ye2018communication, halbawi2017improving, ferdinand2018anytime, li2018near, yu2018lagrange, karakus2017straggler, ozfaturay2018speeding, kiani2018exploitation, chen2018draco, maity2018robust, li2016unified, li2016fundamental, yang2019secure, raviv2017gradient}, matrix multiplication, see e.g.,  \cite{KS18, lee2018speeding, BPR17, AF10, yu2017polynomial, wang2018coded, baharav2018straggler, fahim2017optimal,yu2018straggler, mallick2018rateless}, linear transforms, see e.g., \cite{dutta2017coded, yang2017computing, DCG16,wang2018fundamental}, and content download, see e.g., \cite{Emina1,JLS12,HPZR12,WJW15,LSHR17,KSS15}. In gradient-descent applications, a framework called gradient coding is studied \cite{ye2018communication, halbawi2017improving,tandon2017gradient} in which the authors present coding techniques to avoid stragglers and perform a gradient descent update at each iteration, i.e., at each iteration the master observes the gradient evaluated at the whole data matrix $A$. This framework requires the master to replicate the data redundantly to the workers. The amount of redundancy depends on the number of stragglers that the master wants to tolerate. In this work, we restrict our focus to techniques for straggler mitigation in machine learning applications and in particular to stochastic gradient-descent-type algorithms. 

\paragraph{Approximate gradient coding}
The works that are closely related to our work are \cite{horii2019distributed,chen2016revisiting, tandon2017gradient,charles2017approximate, wang2019fundamental, wang2019erasurehead, maity2018robust, halbawi2017improving,ye2018communication}. In \cite{chen2016revisiting} the authors require the workers to sample a subset of the data from the master. Each worker computes one update of the gradient and sends the result to the master. The master waits for the fastest $n-s$ workers, for a given $s<n$, and performs an update on $\bbeta$. This method is the closest to the Ignore--Stragglers--SGD algorithm that we discussed above. The difference is that in \cite{chen2016revisiting} the workers sample a different subset at each iteration, whereas in Ignore--Stragglers--SGD the workers are given a partition of the data that is fixed throughout the algorithm. In \cite{maity2018robust}, the authors focus on linear loss functions and use LDPC codes to encode the data sent to the workers. If fewer than $s$ stragglers are present, then the master can compute the exact gradient. However, if more than $s$ stragglers are present the master leverages the LDPC code to computes an estimate of the gradient. In \cite{horii2019distributed} the authors propose to distribute the data to the workers using an LDGM code. 
The main drawback is that the master has to run a decoding algorithm to decode the sum of the partial gradients at each iteration. In \cite{tandon2017gradient,charles2017approximate}, the authors present coding frameworks that trade redundancy for computing the exact gradient with high probability. The main idea is to show how far the computed gradient is from the actual gradient as a function of the redundancy factor, the distance between the actual gradient and the computed gradient is termed as error. In \cite{tandon2017gradient} the authors present a data distribution scheme based on Ramanujan graphs. In \cite{charles2017approximate} the authors present two constructions.  The first is based on fractional repetition codes (FRC) and partitions the workers and data into blocks; within a block, each worker receives every data point from the corresponding block.  The second construction called \em Bernoulli Gradient Coding \em (BGC) distributes each data point randomly to $d$ different workers.  We note that BGC is an approximation of the pairwise-balanced schemes we consider and can be seen as a case of SGC when all the data $\mathbf{a}_i$ have the same norm.
In \cite{wang2019fundamental} the authors present fundamental bounds on the error as function of the redundancy. In \cite{wang2019erasurehead}, the authors analyze the convergence rate of the fractional repetition scheme presented in \cite{charles2017approximate} and show that under standard assumptions on the loss function, the algorithm maintains the convergence rate of centralized stochastic gradient descent.

\paragraph{Other work on stochastic gradient descent}
Beginning with its introduction in \cite{RM51}, there has been a huge body of work on stochastic gradient descent (in a setting without stragglers), and we draw on this mathematical framework for our theoretical results.  In the special case of $\ell_2$ loss (which we focus on in this work), SGD coincides with the randomized Kaczmarz method~\cite{SV09,NWS14}, and our proof of Theorem~\ref{thm:convl2} is inspired by these analyses.  

There has been a great deal of work on SGD and Batch--SGD in distributed settings. We discussed some of it above; most of the remaining work focuses on schemes which do not add any redundancy between the workers, see .eg.,~\cite{CSSS11,AD11,DGSX12,SS14, NW16}. In addition to the synchronous setting in which the master waits for all workers to make an update on $\bbeta$, there has been lots of work on the asynchronous setting in which the master makes an update on $\bbeta$ every time a worker gets back, see e.g.,~\cite{gimpel2010distributed,BT89,dean2012large,dutta2018slow}. The tradeoff between synchronous and asynchronous SGD is in the time per iteration versus number of iterations till convergence. Synchronous SGD requires less number of iterations to converge, however due to the synchrony between workers, each iteration takes a longer time. On the other hand, in asynchronous SGD, every time a worker finishes computing the gradient at hand, it sends the result to the master who updates $\bbeta$ and sends its new value to this worker to start a new computation. The problem of asynchronous SGD is that each worker is operating on a different value of $\bbeta$ which may be very old compared to the value of $\bbeta$ held at the master which creates problems in the convergence. Theoretical convergence of asynchronous SGD is hard to analyze by itself, however researchers try to compare the behavior of asynchronous SGD to that of synchronous SGD. For example, \cite{CDR15} shows that asynchronous SGD asymptotically behaves similarly to synchronous SGD in terms of convergence for convex optimization and under the similar assumptions on the loss function. In \cite{dutta2018slow}, the authors compare the convergence rate of synchronous and asynchronous SGD as a function of the wall clock time rather than number of iterations. The authors show that in the beginning of the algorithm asynchronous SGD is faster than synchronous SGD, but gets slower as the algorithm evolves.